\documentclass[11pt, letterpaper]{article}
\usepackage[margin=0.65in]{geometry}
\usepackage[
    backend=biber,
    style=numeric,
    natbib=true,
    backref=true,
    giveninits=true,
    uniquename=init
]{biblatex}
\usepackage{enumitem}
\usepackage{graphicx}
\usepackage{amsmath,amsthm,amssymb,mathtools}
\usepackage{array}
\usepackage{booktabs}
\usepackage{pifont}
\usepackage{paralist}
\usepackage{longtable}
\usepackage{colortbl}
\makeatletter
\@ifundefined{insert@pcolumn}{\let\insert@pcolumn\insert@column}{}
\providecommand\UseTaggingSocket[1]{}
\providecommand\CT@tbl@gdecr@row@count{}
\makeatother
\usepackage{hhline}
\usepackage{bm}
\usepackage{xspace}
\usepackage{url}
\usepackage{boxedminipage}
\usepackage{wrapfig}
\usepackage{subfigure}
\usepackage{ifthen}
\usepackage{color}
\usepackage{xcolor}
\usepackage{framed}
\usepackage{algorithm}
\usepackage{algpseudocode}
\usepackage{float}
\usepackage{rotating}

\usepackage[colorlinks=true,pdfpagemode=UseNone,urlcolor=blue,linkcolor=blue,citecolor=violet,pdfstartview=FitH]{hyperref}

\newcommand{\interior}[1]{\textrm{int}}

\newcommand{\ignore}[1]{}

\newcommand{\cD}{\mathcal{D}}

\newcommand{\cG}{\mathcal{G}}

\newcommand{\CC}{\mathbb C}

\newcommand{\eps}{\varepsilon}

\newcommand{\poly}{\mathrm{poly}}

\newcommand{\EX}{\hbox{\bf E}}

\newcommand{\Sec}[1]{\hyperref[sec:#1]{\S\ref*{sec:#1}}} 
\newcommand{\Eqn}[1]{\hyperref[eq:#1]{(\ref*{eq:#1})}} 
\newcommand{\Fig}[1]{\hyperref[fig:#1]{Fig.\,\ref*{fig:#1}}} 
\newcommand{\Tab}[1]{\hyperref[tab:#1]{Tab.\,\ref*{tab:#1}}} 
\newcommand{\Thm}[1]{\hyperref[thm:#1]{Theorem\,\ref*{thm:#1}}} 
\newcommand{\Fact}[1]{\hyperref[fact:#1]{Fact\,\ref*{fact:#1}}} 
\newcommand{\Lem}[1]{\hyperref[lem:#1]{Lemma\,\ref*{lem:#1}}} 
\newcommand{\Prop}[1]{\hyperref[prop:#1]{Prop.~\ref*{prop:#1}}} 
\newcommand{\Cor}[1]{\hyperref[cor:#1]{Corollary~\ref*{cor:#1}}} 
\newcommand{\Conj}[1]{\hyperref[conj:#1]{Conjecture~\ref*{conj:#1}}} 
\newcommand{\Def}[1]{\hyperref[def:#1]{Definition~\ref*{def:#1}}} 
\newcommand{\Alg}[1]{\hyperref[alg:#1]{Alg.~\ref*{alg:#1}}} 
\newcommand{\Clm}[1]{\hyperref[clm:#1]{Claim~\ref*{clm:#1}}} 
\newcommand{\Obs}[1]{\hyperref[obs:#1]{Observation~\ref*{obs:#1}}} 
\newcommand{\Rem}[1]{\hyperref[rem:#1]{Remark~\ref*{rem:#1}}} 
\newcommand{\Con}[1]{\hyperref[con:#1]{Construction~\ref*{con:#1}}} 
\newcommand{\Step}[1]{\hyperref[step:#1]{Step~\ref*{step:#1}}} 
\newcommand{\Assumption}[1]{\hyperref[assm:#1]{Assumption\,\ref*{assm:#1}}} 

\makeatletter

\newcommand{\ProbabilityRender}[2]{
  \@ifnextchar\bgroup%
  {\renderwithdist{#1}{#2}}
   {\singlervrender{#1}{#2}}
}
\newcommand{\singlervrender}[2]{%
   \ensuremath{\mathchoice
       {{#1}\left[ #2 \right]}
       {{#1}[ #2 ]}
       {{#1}[ #2 ]}
       {{#1}[ #2 ]}
   }
}
\newcommand{\renderwithdist}[3]{%
   \@ifnextchar\bgroup
   {\superfancyrender{#1}{#2}{#3}}
   {\ensuremath{\mathchoice
      {\underset{#2}{#1}\left[ #3 \right]}
      {{#1}_{#2}[ #3 ]}
      {{#1}_{#2}[ #3 ]}
      {{#1}_{#2}[ #3 ]}
     }
   }
}
\newcommand{\superfancyrender}[5]{
   \ensuremath{\mathchoice
      {\underset{#1}{{#1}}\left#4 #3 \right#5}
      {{#1}_{#2}#4 #3 #5}
      {{#1}_{#2}#4 #3 #5}
      {{#1}_{#2}#4 #3 #5}
   }
}

\newcommand{\mainproc}{\texttt{CohortRecovery}}
\newcommand{\contain}{\texttt{bag}}

\newtheorem{theorem}{Theorem}[section]

\newtheorem{lemma}[theorem]{Lemma}
\newtheorem{claim}[theorem]{Claim}
\newtheorem{definition}[theorem]{Definition}

\newcolumntype{H}{>{\setbox0=\hbox\bgroup}c<{\egroup}@{}}

\title{Finding Many Overlapping Dense Subgraphs Using Triadic Cohorts}

\author{
  Sabyasachi Basu \\
  Microsoft Research \\
  \texttt{sabyasachi.basu@microsoft.com}
  \and
  C. Seshadhri \\
  UC Santa Cruz \\
  \texttt{sesh@ucsc.edu}
}
\date{}

\begin{document}
\maketitle
\begin{abstract}
Graphs are a standard representation for data in the social sciences, cybersecurity, computer infrastructure,
bioinformatics, and more. Typical real-world graphs are sparse,
meaning the average degree is small (in the tens, while the number of vertices is more than
millions). When graph data is collected from a source, a major task is to perform data exploration.
Thus, any region of ``density" is of interest, since it indicates special structure.

An important goal is to cover a significant portion of the graph using dense subgraphs. Existing algorithms
that find many dense subgraphs do not have overlapping output, and hence provide limited coverage. Other
methods that produce overlapping clusters do not generate dense subgraphs. The main goal of this paper
is to develop provable and practical methods that can provide overlapping dense subgraphs that can cover
large portions of real-world networks.

Our contribution is an algorithm \mainproc{} that achieves this goal. We first develop a mathematical
framework of \emph{triadic cohorts} that captures the notion of ``detectable dense subgraphs'' that potentially
overlap. We prove that \mainproc{} can output a set of dense subgraphs, such that each triadic cohort is
almost completely contained in some dense subgraph. We give
a practical implementation of \mainproc{} and demonstrate it on a variety of datasets. It typically runs
in under ten minutes on a commodity machine even on graphs with tens of millions of edges. For numerous datasets, \mainproc{} is able to cover more than 25\% of the vertices
in non-trivial subgraphs of density more than 0.8, and is significantly better than a wide variety of scalable graph clustering/community detection
algorithms. Moreover, we demonstrate that output of \mainproc{}
captures ground truth clusters obtained by manual curation.
\end{abstract}
		
\maketitle

\section{Introduction}

The explosive growth of network science as a discipline has been fueled by the wide
variety of datasets that can be represented as large graphs. 
When graph data is collected from a source, a major task is to perform data exploration, often in the context
of unsupervised machine learning.
Real-world graphs typically have low average degree (in the tens), while having
millions or more vertices.
Thus, any region of ``density'' usually indicates special or uncommon structure.
This problem is often given the umbrella name of \emph{dense subgraph discovery}.
It has been used for biological data extraction~\cite{ZH05,FNBB06}, graph compression~\cite{BC08},
detecting social groups~\cite{Beal2003CohesionAP,forsyth2010group}, detecting link farms~\cite{KRRT99, DGP07, GKT05},
graph visualization~\cite{AhDBV05}, financial data analysis~\cite{DJDLT09}, and 
improving the throughput of servers~\cite{GJLSW13}.
This topic has been studied intensely from the early days of graph mining~\cite{LeRu+10}.

Let $G = (V,E)$ denote an input undirected, unweighted graph, and $E(S)$ denote the number of edges
in a subset $S$ of vertices.
We define the \emph{density} of $S$ as $E(S)/\binom{|S|}{2}$.
Since most real-world graphs have densities well below $10^{-4}$,
a set with density of (say) $0.5$ is of interest. There are many formal objectives
for the largest dense set, or the densest set with a size threshold~\cite{Gol84,AndC09}.
These objectives, even with approximations, are typically NP-hard~\cite{Håstad1999, Fe02, Khot06}.
There is a rich literature of practical heuristics~\cite{Charikar, GKT05, AndC09, Flowless, Ts15, Tsou14, SEF16, KS22, DCS17, Sot20, BGMT21, CQT22}
(refer to recent survey~\cite{LMFB23}). 
A drawback of most methods is that they are geared towards finding one (or a few) large dense subgraphs.
From a data analysis perspective, one 
wants to recover information about many dense subgraphs instead of just a single ``optimal''.
In such settings, one wishes to cover large portions
of the graph using dense subgraphs, to explain or get higher level representations of the 
input. Finding just a few optimally densest sets has limited utility, since vast portions of the graph cannot be
covered. In general, sets of high density typically have size in order of tens~\cite{SaSePi+15,Ts15,BaPeQi+24},
so we want to find \emph{thousands} of such dense sets.

There are recent results on this problem~\cite{GKT05, SaSePi+15,SEF16,Sot20, BGMT21, BaPeQi+24}. 
But a crucial aspect ignored by many results (notable exceptions are~\cite{AGM12,XLP+25}) is that of \emph{overlap}. Methods that find
many dense components in a graph usually construct partitions. One expects
real-world data to have significant overlap in dense clusters. For example, we note that in many datasets with ground truth clusters~\cite{snapnets}, more than 30\% of the graph belong to multiple clusters (refer to  \Tab{gt-membership}). 

Our motivating problem is: \emph{can we design a scalable algorithm that can cover large
portions of real-world graphs with overlapping dense subgraphs?} The primary goal
is to maximize the covered portion, potentially exploiting density to increase that coverage.

To give our problem and approach a theoretical grounding, we wish to
build a mathematical framework of recovering overlapping dense
subgraphs, that leads to practical algorithms. Our ideal aim is to formally show that 
all ``well-structured'' dense subgraphs can be approximately recovered, and demonstrate the
algorithm empirically. 
\begin{figure*}
    \centering
    \includegraphics[width = 0.46\linewidth]{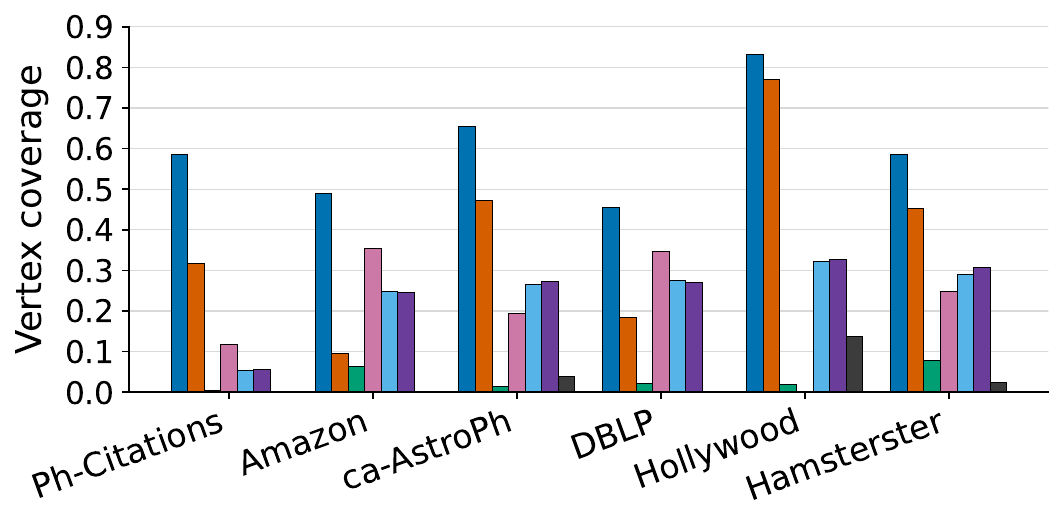}
    \includegraphics[width = 0.46\linewidth]{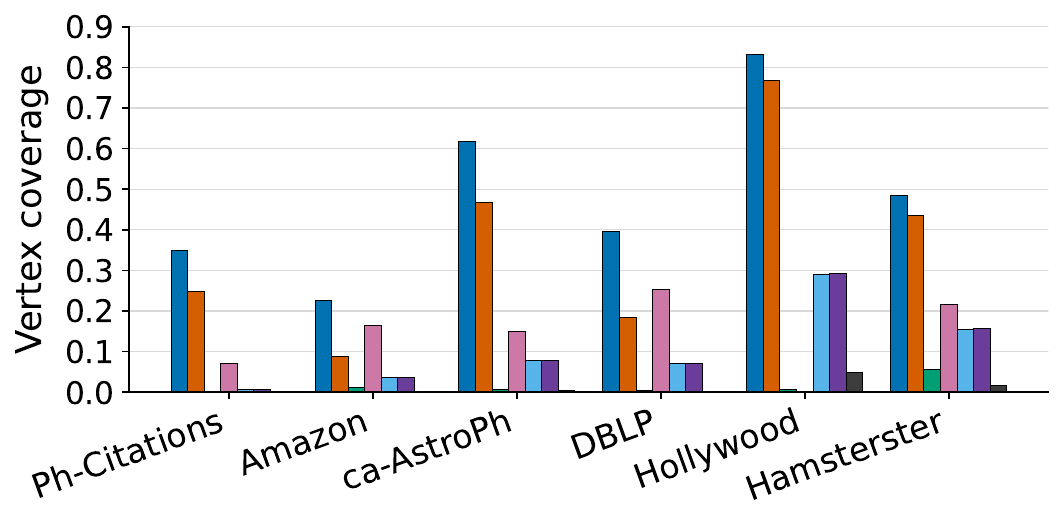}
    \includegraphics[width = 0.6\linewidth]{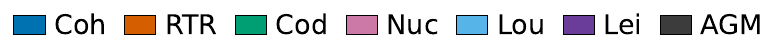}
    \caption{For an (overlapping or otherwise) clustering method, the \emph{coverage} is the fraction of vertices contained in dense clusters of size at least 5. 
    We focus on density at least 0.5 (left) and 0.8 (right). Our algorithm \mainproc\ is labeled as Coh. We compare to RTREx (RTR~\cite{BaPeQi+24}), CoDeSEG (Cod~\cite{XLP+25}), Nucleus (Nuc~\cite{SaSePi+15}), Louvain (Lou~\cite{Louvain}), Leiden (Lei~\cite{Leiden}), and a PageRank technique (AGM~\cite{AGM12}). Across the datasets, we see that the
    CohortRecovery algorithm has a better coverage, and in many cases, has 10\% more coverage than any competing algorithm. We also observe that no previous algorithm
    has good performance on all datasets.}
    \label{fig:coveragebar}
    \vspace{-.5cm}
\end{figure*}
\subsection{Our Contributions} \label{sec:contributions}

We define a new formulation of \emph{triadic cohorts}, that captures a notion of 
recoverable dense substructures. We design an algorithm \mainproc{} that can provably
discover all such structures. \mainproc{} has excellent empirical behavior, and can cover
large fractions of real-world graphs by dense subgraphs. We detail our contributions below.

\begin{asparaitem}
\item {\em Triadic cohorts:} There is a rich history of results exploiting triangle
    structure for dense subgraph discovery~\cite{SaSePi+15,Ts15,BeGlLe16,BBS24,BaPeQi+24}. Our theoretical framework 
    defines a notion of triangle-rich dense subgraphs called triadic cohorts. The 
    main problem is to extract an approximate cover that contains \emph{every} possible
    triadic cohort. We give a provable algorithm \mainproc{} for this problem that runs in polynomial time.
    (The running time is subquadratic in the triangle count.) While there are many existing notions of 
	dense subgraphs, the cohorts notion allows for (nearly exact) recovery of all cohorts in polynomial time.
\item {\em Practical implementation:} Our main practical contribution is an
 implementation of \mainproc. We do a detailed evaluation on a large number of datasets.
    \mainproc\ typically ran in minutes, and always
    terminated within an hour. In our largest dataset, the \texttt{orkut} network with 117 million edges, \mainproc\ took less than eight minutes on commodity
    hardware. 
\item {\em Excellent coverage:} The true test of \mainproc{} is the ability to cover
    a graph with large, dense subgraphs. The output of \mainproc{} is a collection of subsets
    of vertices. Consider \Fig{coveragebar}. We compute the \emph{coverage} of these subsets,
	which is the fraction of the vertices present in subsets of sufficiently high density.
	(We remove all subsets of size at most four, to prevent false coverage by small sets.)
    We compute the coverage of dense subgraphs of density at least $0.5$ (respectively, $0.8$) and
    at least five vertices. We compare with a large suite of existing algorithms that use triangle
	information (Nucleus, RTR)~\cite{SaSePi+15, BaPeQi+24}, modularity maximization (Louvain, Leiden)~\cite{Louvain, Leiden},
	and methods that explicitly find overlapping clusters (PageRank methods, CoDeSEG~\cite{AGM12, XLP+25}).
	(We explain more in related work, but graph embedding methods and Graph Neural Networks typically do not scale
	to our data sizes, and require significant cost and specialized hardware~\cite{WLB+20,TPPM23,Goel2025}.)	
    We experiment on a large collection of 18 datasets (full data in~\Tab{cohorts-coverage}, but we
	show representative results in \Fig{coveragebar}).
    The coverage of \mainproc{} is the highest in almost all the 18 datasets, and no previous algorithm
	performs well on the coverage metric for all datasets. In numerous data sets, such as livejournal, large-dblp, and Berkstan,
    the coverage of \mainproc{} is nearly double of the previous methods. Though previous methods
	were not necessarily optimized for coverage, but we believe that coverage is a fundamental metric for a dense
	subgraph discovery algorithm.

\item {\em Ground truth comparisons:} We run all the methods on datasets which have curated ground truth clusters, based on network metadata like node category, or social network groups.
(We acknowledge that this data is imperfect, and is mostly a guide on relevant parts of the network ~\cite{ground}.) 
    The output of \mainproc{} captures a significant fraction of the ground truth clusters. We evaluate this by computing the fraction of an output cluster
that is contained in a ground truth cluster. We observe that \mainproc{} performs the best in all but one of the datasets (\Tab{gt-bc} in \Sec{findings}). In numerous datasets, \mainproc{} had an improvement of 7 percent over previous methods. In addition, we consider simple Stochastic Block Models with and without overlap, to validate that \mainproc{} can recover
the ground truth structure. We see that \mainproc{} has near perfect recovery, while other methods typically fail when there are overlapping blocks (\Tab{sbm-app} referenced in \Sec{findings}).
\end{asparaitem}

\subsection{Related Work}
The problem of dense subgraph discovery has been studied by a diverse group of researchers in both theory and applied communities. Most reasonable formulations are NP hard, and a lot of research has gone into approximation algorithms. We refer to a number of surveys on dense subgraph discovery for a discussion on the theory~\cite{Fang20, LMFB23}.

Dense subgraph discovery is also extremely well-studied in the data mining community, where a series of results have demonstrated a suite of algorithms for this task~\cite{Ts15, Flowless, denser, KS24} leveraging cliques, quasi-cliques, as well as core and truss based decompositions ~\cite{AndC09, WZTT10, WC12, Huang_2017, XMFB23}. An ever-growing series of results also use triangles to find dense clusters in large networks ~\cite{SaSePi+15,Ts15,BeGlLe16,TPM17, VGW18, KS20, KS22}. We note that algorithms studied in the community detection literature, such as the Louvain~\cite{Louvain} and Leiden~\cite{Leiden} algorithms often perform well at this task; we point the reader to the survey by Jin et al.~\cite{Jin_2021} for a detailed look at different community detection methods. Overlapping clusters are far less well understood. Only a handful of results systematically study this~\cite{Khandekar2013,AGM12,Gargi_Lu_Mirrokni_Yoon_2021, Dondi2021, GGT16, He2023-xn}, and even fewer provide scalable implementations that work for real world networks. The work of Yang and Leskovec~\cite{YL12, ground} has attempted to model community overlaps in real-world networks. Overlapping subgraph recovery has been approached in more recent work as well, such as the information theoretic approaches of Xian et al.~\cite{XLP+25}. A separate line of work attempts to produce overlapping clusters by leveraging maximal cliques in graphs~\cite{WCL+17,CliquePercolation,wen2016maximal, BBC+24, QC}. However, maximal clique enumeration and its extensions are well known NP-hard problems, and their relaxations continue to be expensive in practice. Moreover, such approaches invariably produce many more clusters, and introduce massive redundancy in terms of the number of clusters a vertex participates in.

There is a rich body of work that studies the degree densest subgraph problem and variants of it with size constraints~\cite{Charikar, AndC09, Flowless, Ts15, Tsou14, KS22, DCS17, CQT22, denser}. 
Past work~\cite{BaPeQi+24} has shown that iterative versions of these techniques to output multiple dense clusters typically does not work, and is also slow on real-world datasets.

In the last decade, graph neural networks (GNNs) have become a popular approach for learning over graph-structured data using end-to-end differentiable objectives. While these methods have achieved impressive performance on many supervised tasks, their application to unsupervised graph clustering remains less clear. 
Tsitsulin et al.~\cite{TPPM23} demonstrate that classical clustering techniques outperform modern GNN-based methods for graph clustering. Moreover, existing GNN-based community detection approaches often require substantially greater computational resources, may require the number of communities as an input or model parameter~\cite{WLB+20}, and can be sensitive to perturbations of the underlying graph~\cite{Goel2025}. Often, these methods only work for graphs with less than million vertices.
As a result, we do not compare with GNN methods for the following reasons: (i) lack of scalability to large graphs,
(ii) the need for specialized hardware for large graphs, and (iii) the results showing poor clustering performance of GNN methods.

The most relevant to our work is the recent result by Basu et al.~\cite{BaPeQi+24}, that uses local triangle information to design a dense subgraph discovery algorithm with good real-world performance. 
Older theoretical work gives justifications for this approach~\cite{BBS24, GuRoSe14}. Our main insight is that these triangle based techniques can be generalized to detect overlapping cohorts without the computational cost of clique enumeration, while retaining density guarantees. 

\section{The problem setup}\label{sec:setup}

Given an input \emph{undirected and simple} graph $G = (V,E)$, we wish to find (potentially) overlapping dense subgraphs that cover a large fraction of the graph.
We set up a theoretical framework. In this framework, we will define a notion of \emph{triadic cohorts}, that attempts
to capture a "ground truth dense substructure" of interest. From an exploratory data analysis perspective, our aim
would be to efficiently find close enough approximations to each of these cohorts. Consider that data has noise
and no definition of substructure is perfect, we would like to find dense subgraphs that "cover" that cohorts
(and no necessarily output them exactly).

In our setup, the parameters $\eps, \gamma$ should be thought of as constants. 
Our implementation uses a single parameter $\eps$, and we do an experimental analysis
on the effect of this parameter. (More discussion in \Sec{prac}.)
We use the notation $O_{\eps,\gamma}(\cdot), \Omega_{\eps,\gamma}(\cdot)$
to hide dependencies of the form $\poly(\eps,\gamma)$. The edge density of a set $S$ is $E(S)/\binom{|S|}{2}$, where $E(S)$ is the set of edges contained in $S$.

We use $d_u$ to denote the degree of vertex $u$, in the original graph $G$. Our algorithm will continually prune vertices
and construct subgraphs, but $d_u$ always denotes the original degree. This is crucial for our algorithm and analysis
for the discovery of cohorts.

Recall that a triangle is clique of size $3$, and numerous results have established the utility of finding triangle-rich substructures~\cite{SaSePi+15,Ts15,BeGlLe16,TPM17,BaPeQi+24}.
We build on the notion of triangle density for defining cohorts. We capture three well-established properties
of a "good" dense subgraph. (i) All edges participate in sufficiently many triangles. (ii) The diameter of the subgraph is $2$.
(iii) Internally, the subgraph should not have small cuts (alternately, it is a good expander).

\begin{definition} \label{def:cohort} Let $\eps \in (0,1)$. A \emph{$\eps$-triadic cohort} is a connected subgraph $H = (V(H),E(H))$ 
satisfying the following conditions.
\begin{asparaenum}
    \item (Low diameter) The diameter of $H$ is at most $2$. (Equivalently, every pair of vertices in $V(H)$ have a common neighbor through $H$.)
    \item (Triangle rich) Every edge $(u,v) \in E(H)$ is part of at least $\eps \max(d_u, d_v, |V(H)|)$ triangles in $H$.
    \item (Expansion) Every subset $S \subseteq V(H)$ of size at most $|V(H)|/2$ has at least $\eps \sum_{v \in S} d_v$
edges (in $H$) going to $V(H) \setminus S$.
\end{asparaenum}
\end{definition}

Note that triadic cohorts can overlap and potentially even contain each other.
To ensure that triadic cohorts are identifiable, we need to define vertices that are specific to one cohort.

\begin{definition} \label{def:strong}  A vertex $v$
is a \emph{sole} member of an $\eps$-triadic cohort $H = (V(H), E(H))$ if $v \in H$, and for
all edges $(u,v)$, the edge $(u,v)$ participates in at most $(\eps/2)\cdot\max(d_u, d_v)$ triangles
where the third vertex is \emph{not} in $V(H)$.

A triadic cohort is called \emph{identifiable} if it contains a sole vertex.
\end{definition}

Observe that the identifiable condition is quite weak, since we only require
\emph{one} vertex to be a sole member.

Our main theorem follows.

\begin{theorem} \label{thm:main-contain} There is an algorithm \mainproc$(G,\eps,\gamma)$
that outputs a family $\CC$ of sets with the following properties.
\begin{asparaenum}
	\item Every set $C \in \CC$ has edge density at least $\Omega_{\eps,\gamma}(1)$.
	\item For \emph{every} identifiable $\eps$-triadic cohort $H = (V(H), E(H))$,
there exists some $C \in \CC$ containing a $(1-\gamma)$-fraction of $V(H)$. Moreover, $|C| = O_\eps(|V(H)|)$.
\end{asparaenum}
\end{theorem}

Note the strength of this theorem. There is no assumption on the structure, overlap, or distribution
of the various cohorts. For the sake of discussion, let us assume $\eps$ is a constant. Each
of the output sets have constant edge density, so \mainproc{} generates dense subgraphs.
\emph{Every} identifiable cohort is approximately recovered, since it is almost completely contained
in some output set that is only a constant factor larger. This is an extremely strong guarantee
given the lack of any assumptions on the underlying graph $G$.

\subsection{Preliminaries} \label{sec:prelims}

Before describing the algorithm \mainproc, let us set up preliminary lemmas. First,
we state some basic properties of triadic cohorts. We use the notation $N(v,H)$
to define the set of neighbors of $v$ in the cohort (or subset) $H$.

\begin{lemma} \label{lem:cohort} Consider an \emph{$\eps$-triadic cohort} $H = (V(H),E(H))$.
Let $u, v$ be any two vertices in $V(H)$.
\begin{asparaitem}
	\item $d_v \leq d_u/\eps^2$.
    \item $N(v ,H) \geq \eps |V(H)|$ (hence $|V(H)| \leq d_v/\eps$)
	\item The edge density of $H$ is $\Omega(\eps)$.
\end{asparaitem}
\end{lemma}

\begin{proof} Consider any edge $(x,y)$ of the cohort. It can participate
in at most $\min(d_x, d_y)$ triangles. By the cohort property, $\min(d_x, d_y) \geq \eps \max(d_x, d_y, |V(H)|)$.
So $d_x \geq \eps d_y$ and $d_y \geq \eps d_x$. Since the diameter is $2$, that implies
for any two vertices $u, v$ in the cohort, $d_v \leq d_u/\eps^2$. 

    Also, the inequality above implies that $d_x \geq \eps |V(H)|$, for any $x$ of non-zero degree in $H$.
Since $H$ is connected, the bound holds for all vertices in $H$. The edge $(x,y)$
participates in at least $\eps \max(d_x,  d_y, |V(H)|)$ triangles in $H$. Hence, $x$
must have at least $\eps \max(d_x, |V(H)|)$ neighbors in $H$. This proves the second point.

The edge density bound follows, since the number of edges in $H$
is at least $(1/2) \sum_{x \in H} N(x,H) \geq (1/2) \eps |V(H)| \sum_{x \in H} 1 \geq (1/2) \eps |V(H)|^2$.

\end{proof}

We will use the classic notion of core decompositions.

\begin{definition} \label{def:kcore} A $k$-core of a graph is a maximal subgraph whose minimum degree is at least $k$.
\end{definition}

Matula-Beck gave a classic linear time algorithm ~\cite{MB83} for computing $k$-cores (for all $k$). The following
lemma is a well known fact about core decompositions. Proof in Appendix~\ref{app:prelims}.

\begin{lemma} \label{lem:kcore} Suppose a graph contains a subgraph with average degree at least $r$.
    Then it contains a $k$-core with $k \geq r/2$.
\end{lemma}

The following is a standard ``reverse" of the classic Markov inequality.
Proof in Appendix~\ref{app:prelims}.

\begin{claim} \label{clm:reverse} [Reverse Markov inequality] Consider a positive
random variable $X$ with maximum value at most $M$. Then $\Pr[X > \EX[X]/2] > \EX[X]/2M$.
\end{claim}

\section{The algorithm} \label{sec:algorithm}

\begin{algorithm}[ht]
	\caption{\mainproc$(G,\eps, \gamma)$: (We use $d_v$ to denote the degree of $v$ in the input graph $G$.)}
	\label{alg:main}
	\begin{algorithmic}[1]
		\State Initialize subgraph $G'$ to $G$.
        \While{there is edge $e = (u,v)$ in $< \eps \max(d_u, d_v)$ triangles in $G'$} \label{step:clean}
			\State Delete $e$ from $G'$
		\EndWhile
        \State Delete all isolated (degree zero) vertices from $G'$
		\State Construct a maximal independent set $I$ of $G'$
		\For{each $i \in I$}
			\State Run \contain$(i,G',\eps,\gamma)$, and add the output to $\mathcal C$ \label{step:call}
		\EndFor
		\State Output $\mathcal C$
	\end{algorithmic}
\end{algorithm}

\begin{algorithm}[ht]
	\caption{\contain$(i,G',\eps,\gamma)$: (We use $d_i$ to denote the degree of $i$ in the original $G$.)}
	\label{alg:clean}
	\begin{algorithmic}[1]
        \State Let $G''$ be the subgraph induced by vertices of degree at most $d_i/\eps^3$.
        \State Initialize $C = \{i\} \cup N(i,G'')$. \label{step:nbr}
        \State Set $\delta = \eps^c \gamma$ (for some explicit constant $c$).
        \State flag = \emph{True}
        \While{flag} \label{step:extend}
           \State flag = \emph{False}
           \State{If there is a vertex $v$ in $G''$ that is in $\delta |C|^2$ triangles with both other vertices in $C$,
            add $v$ to $C$ and set flag to be \emph{True}.} \label{step:vertex}
            \For{each edge $(x,y)$ of $G''$ that is in $\delta|C|$ triangles with the third vertex in $C$}
                \State{Color edge $(x,y)$ red} \label{step:edge}
            \EndFor
            \State Compute the core decomposition of the subgraph of red edges. \label{step:core}
            \State If there exists a core with minimum degree at least $\delta |C|$, add all these vertices
            to $C$ and set flag to be \emph{True}. \label{step:core-ext}
		\EndWhile
		\State Output $C$
	\end{algorithmic}
\end{algorithm}

We give a high level exposition of the algorithm. The first step in \mainproc{} is a "cleaning" operation,
wherein edges with few triangles are removed from the graph. The operation has been observed
to be useful for dense subgraph discovery, probably first observed in~\cite{SaPa11}, and
used in many other results~\cite{SaSePi+15,Ts15,BeGlLe16,TPM17,BaPeQi+24}. 

The next step is to choose a set of seeds, from which we extract dense subgraphs. Our idea
is to use a maximal independent set $I$. Intuitively, we want to choose seeds sufficiently
far away from each other, but still give enough coverage of the graph.
The surprising part of the analysis is that such an independent set suffices to capture
\emph{all} triadic cohorts, despite no assumption of the intersection patterns of these cohorts.

The real work happens in \contain$(i,G',\eps,\gamma)$. The obvious first step to initialize the output cluster $C$
is the neighborhood of $i$ in $G''$. (It is easiest to think of $G''$ as $G'$. For technical reasons, $G''$ is defined as
the subgraph on vertices of comparable degree. In practice, this is not needed, but it makes the analysis
easier.) The challenge is to now grow this neighborhood into a dense subgraph that contains triadic cohorts.
The key idea, again common to many previous results, is to absorb as many triangles as possible into $C$.
Roughly speaking, the size of $C$ will be $O(d_i)$. So every vertex added to $C$ must bring in $\Omega(d^2_i)$
triangles, to ensure that triangle density is maintained. That explains \Step{vertex}. 

This leads to the difficult case and one of the main innovations of our theory. There can be many triangles
that have a single vertex in $C$ and two vertices outside. Such triangles cannot be absorbed by \Step{vertex}.
We need to bring in \emph{edges} outside $C$ that participate with many triangles with their third vertex in $C$
(these are denoted red in \Step{edge}). Yet, adding each such edge to $C$ would only bring in $O(d_i)$ triangles,
while adding two vertices. A naive growing process of adding edges would dilute the triangle density.

We observe that the "right" operation is to find vertices that are incident to many red edges. And a subgraph
of such vertices is exactly a $k$-core of red edges. So we compute core decompositions and directly add large
enough $k$-cores to $C$.

This algorithm introduces many analysis challenges. Firstly, we must argue that this iterative
process of growing $C$ converges with $C$ ending as a dense (and triangle dense) structure. Secondly,
we need to prove that triadic cohorts are almost entirely contained in some output dense cluster. This is quite challenging, since we need to argue the neighborhoods explored capture almost the entire cohort.

\subsection{The main analysis} \label{sec:contain}

\begin{claim} \label{clm:cohort-edge} The subgraph $G'$ obtained after the \emph{while} loop (starting in \Step{clean}) of \mainproc$(G,\eps)$
contains all edges of all $\eps$-triadic cohorts. \end{claim}

\begin{proof} Consider an $\eps$-triadic cohort $H = (V(H),E(H))$. Suppose $G'$
does not contain all edges of $E(H)$. Let $(u,v)$ be the first edge from $H$
deleted in the \emph{while} loop of \Step{clean} on \mainproc. Since $(u,v)$ participates
in at least $\eps \max(d_u, d_v)$ triangles in $H$, just before deletion of $(u,v)$,
it participates in at least the same number of triangles in $G'$. But this means
that it cannot be deleted, a contradiction. Hence, no edges of a cohort can be 
deleted during the construction of $G'$.
\end{proof}

The following claim uses similar arguments, so we give the proof
in Appendix~\ref{app:contain}.

\begin{claim} \label{clm:i-degree} Consider the subgraph $G''$ in \contain$(i,G',\eps,\gamma)$.
The vertex $i$ has degree at least $\eps d_i$ in $G''$. Moreover, each edge incident $i$
participates in at least $\eps d_i$ triangles in $G''$.
\end{claim}

We prove that $C$ cannot be too big (compared to the degree of $i$), and
$C$ must be dense. This is related to proving the termination of the growing process.

\begin{claim} \label{clm:size} Let $C$ denote the output of \contain$(i,G',\eps,\gamma)$.
    Then, $|C| = O_{\eps,\gamma}(d_i)$ and the density of $C$ is $\Omega_{\eps,\gamma}(1)$.
\end{claim}

\begin{proof} We first bound the size of $C$, just at the time of termination.
Suppose the last vertex is added (to $C$) through \Step{extend} of \contain. 
Suppose a vertex $v$ was added because it is in $\delta |C|^2$ triangles with
both other vertices in $C$. The vertex $v$ can participate in at most $d^2_v \leq d^2_i/\eps^6$
triangles. Hence, $|C| \leq d_i/(\sqrt{\delta} \eps^3)$ and the final output has size
$|C| + 1$. That completes the proof in this case.

Suppose the last addition to $C$ was a core. Every red edge $(x,y)$ participates in at least $\delta |C|$
triangles, so $\delta |C| \leq d_x \leq d_i/\eps^3$. The set $C$ participates
in at most $|C| d^2_i/\eps^6$ triangles. Since distinct red edges
participate in different triangles with a third vertex in $C$, there are at most $d^2_i/(\delta \eps^6)$
red edges. The core has a minimum (red) degree of $\delta |C|$. If the core has $k$
vertices, it involves at least $\delta k |C|/2$ edges. Hence,
$k \leq 2d_i^2/(\delta^2\eps^6|C|)$.
Since $C$ contains at least the neighborhood of $i$ in $G''$,
$|C| \geq \eps d_i$ (by \Clm{i-degree}). Hence, $k = O_{\eps,\gamma}(d_i)$.
Note that the minimum degree of the core is at least $\delta |C|$ and is at most $k$.
Hence $\delta |C| \leq k = O_{\eps,\gamma}(d_i)$. So the size of $C$,
even after the addition of the core, is $O_{\eps,\gamma}(d_i)$.

So we have proven that $|C| = O_{\eps,\gamma}(d_i)$. Now, we lower bound the edge density. 
By \Clm{i-degree}, every edge incident to $i$ in $G''$
participates in at least $\eps d_i$ triangles in $G''$. In total, $i$ participates
in at least $\eps^2 d^2_i/2$ triangles in $G''$ (triangles might be counted twice). Each such triangle contributes a unique edge
in the neighborhood of $i$. Since all vertices in the neighborhood are part of $C$,
$C$ contains at least $\eps^2 d^2_i$ edges. $|C| = O_{\eps,\gamma}(d_i)$,
so the edge density is $\Omega_{\eps,\gamma}(1)$.
\end{proof}

We pick a specific identifiable $\eps$-triadic cohort $H = (V(H), E(H))$ and let $v$
denote a sole member in $H$. 

\begin{claim} There exists some $i \in I$
    such that at least $(\eps^2/2) |V(H)|$
vertices of $H$ are contained in $C$ as initialized in \Step{nbr}
of the call to \contain$(i,G',\eps,\gamma)$.  \label{clm:goodseed}
\end{claim}

\begin{proof} Since $I$ is a maximal independent set, either $v \in I$
or some neighbor of $v$ (in $G'$) is in $I$. Let $i$
denote that vertex. 

{\em Case 1, $i = v$:} By \Clm{cohort-edge}, all edges of $H$ are in $G'$.
Moreover, by \Lem{cohort}, all neighbors of $i$ in $H$ have degree at most $d_i/\eps^2$.
    Hence, $N(i, G'')$ contains $N(i,H)$. \Lem{cohort} asserts that $N(i, H) \geq \eps |V(H)|$,
completing the proof in this case.

{\em Case 2, $i \neq v$:} Since $(i,v)$ is an edge in $G'$, it participates in at least
$\eps\max(d_i,d_v)$ triangles in $G'$. Since $v$ is a sole member, at least $(\eps/2)\max(d_i,d_v)$ of these triangles have their third vertex in $V(H)$. Every such third vertex $w$ belongs to $G''$: since $(i,w)\in E(G')$, the cleaning condition implies $d_w\leq d_i/\eps\leq d_i/\eps^3$. Finally, \Lem{cohort} gives $d_v\geq\eps|V(H)|$. Therefore, $N(i,G'')$ contains at least $(\eps^2/2)|V(H)|$ vertices of $H$.

\end{proof}

The main technical work is done in the next lemma asserting that, for every cohort $H = (V(H),E(H))$, some output set
contains almost all of $H$. We give a high level explanation of the proof. By the previous claim,
there is some starting vertex $i$ such that $C$ contains at least $\eps^2 |V(H)|/2$ vertices of $H$.
The set $C$ will continually grow until the \emph{while} loop in the call \contain$(i,G',\eps,\gamma)$ terminates. Suppose $C$
has less than a $(1-\gamma)$-fraction of $V(H)$. Using the cohort definitions and \Lem{cohort},
we show there are many triangles of $H$ that are cut by $C$. These triangles either
have two vertices in $C$, or exactly one vertex in $C$. In the former case,
by a careful accounting, we prove that there is some vertex \emph{outside} $C$
that has many triangles to $C$, and \Step{vertex} of \contain\ will add it to $C$.
In the latter case, we prove that there must be a dense subgraph of red edges,
so \Step{core-ext} will add to $C$. Overall, the \emph{while} loop does not terminate
as long as $C$ has less than a $(1-\gamma)$-fraction of $V(H)$.

\begin{lemma} \label{lem:contain} For a cohort $H = (V(H), E(H))$, let $i$ be the vertex
obtained in \Clm{goodseed}. The output set of \contain$(i,G',\eps,\gamma)$
    contains a $(1-\gamma)$-fraction of $V(H)$.
\end{lemma}

\begin{proof} For convenience, we use $S$ to denote $V(H)$.
    
    By \Clm{goodseed}, the output set $C$ contains at 
least $(\eps^2/2) |S|$ vertices of $S$. Suppose, for contradiction's sake,
$C$ contained less than $(1-\gamma)|S|$ vertices of $S$.
Then, we will prove that the while loop of \Step{extend} of \contain\ would not terminate; i.e. it continues till more vertices from $S$ are included in the output.

We will assume that $\gamma$ is sufficiently small, compared to $\eps$.
Observe that $S \setminus C$ has size in the range $(\gamma |S|, (1-\eps^2/2)|S|)$.
Hence, $\min(|S \cap C|, |S \setminus C|) \geq \gamma |S|$.
Note that either $|S \cap C|$ or $|S \setminus C|$ is at most $|S|/2$.
By \Lem{cohort}, each vertex in $S$ has degree at least $\eps |S|$.
By the expansion property of cohort (\Def{cohort}), there are at least $\eps (\gamma |S|)(\eps |S|)
= \gamma \eps^2 |S|^2$ edges 
edges of $H$ leaving $C$ (call these cut edges). All edges of $H$ are present in $G'$. Moreover, all
vertices in $H$ have degree at most $d_i/\eps^3$, since $i$ is a neighbor
of some vertex in $H$. Each cut edge participates in at least $\eps |S|$ triangles
in $H$, called cut triangles. Each such triangle involves exactly two such cut edges. So if we count
the number of triangles incident to each cut edge, each cut triangle is counted twice.

Overall, the graph $G''$ contains at least $\gamma \eps^2 |S|^3/2$ cut triangles of $H$, that contain
at least one vertex in $C$. The remaining vertices of these triangles lie in $S$.
We split in two cases: either a majority of these cut triangles have one vertex in $C$,
or they have two.
 
{\em Case 1, there are at least $\gamma \eps^2 |S|^3/4$ cut triangles with exactly two vertices in $C$.}
For every vertex $u$, let $t_u$ be the number of triangles it forms with two vertices in $C$. 
Since the procedure has terminated, $\forall u, t_u \leq \delta |C|^2$. On the other hand,
this case asserts that $\sum_u t_u \geq \gamma \eps^2 |S|^3/4$. We will prove a contradiction.
Firstly, $\forall u, \sqrt{t_u} \geq \sqrt{\delta} |C| / 2$. Secondly, $\sqrt{t_u} \leq b_u$,
where $b_u$ is the number of neighbors that $u$ has in $C$.

\begin{equation}
\gamma \eps^2|S|^3/4 \leq \sum_u t_u = \sum_u \sqrt{t_u} \sqrt{t_u} \leq \sum_u \sqrt{\delta} |C| b_u
= \sqrt{\delta} |C| \sum_u b_u
\end{equation} 

By \Clm{size}, $|C| = O_\eps(d_i)$. By the assumption of $i$ (from \Clm{goodseed}), it is
a neighbor of some vertex in $S$. Hence, $d_i = O_\eps(|S|)$.
$\sum_u b_u = \Omega_\eps(\delta |S|^2/\sqrt{\gamma})$. Note that $\sum_u b_u$
is exactly the sum of degrees of vertices in $C$. Each vertex in $C$ has degree at most $d_i/\eps^3$,
and $|C| = O_\eps(d_i)$. Moreover, $|S| = \Theta_\eps(d_i)$, so $\sum_u b_u = O_\eps(|S|^2)$.
Hence, for sufficiently small $\delta$, we get a contradiction.

{\em Case 2, there are at least $\gamma \eps^2 |S|^3/4$ cut triangles with exactly one vertex in $C$.}
Let $F$ be the set of third (non-cut) edges of these triangles. These are edges of $H$ with both
endpoints in $S \setminus C$. For each $e \in F$, let $X_e$ be the number of these triangles whose
non-cut edge is $e$, and let $T = \sum_{e \in F} X_e$. Thus,
$T \geq \gamma \eps^2 |S|^3/4$. Since every vertex of $G''$ has degree at most $d_i/\eps^3$,
$|F| \leq |S|d_i/(2\eps^3)$, so
\[
\mu := \frac{T}{|F|} \geq \frac{\gamma \eps^5 |S|^2}{2d_i}
= \Omega_\eps(\gamma d_i),
\]
where we use $|S| = \Theta_\eps(d_i)$. Moreover, $X_e \leq M := d_i/\eps^3$ for every $e \in F$.

Apply the reverse Markov inequality of \Clm{reverse} to an edge chosen uniformly from $F$. If $q$
is the number of edges with $X_e > \mu/2$, then
\[
q \geq |F|\frac{\mu}{2M} = \frac{T}{2M}
\geq \frac{\gamma \eps^5 |S|^3}{8d_i}
= \Omega_\eps(\gamma d_i^2).
\]
Each such edge participates in more than $\mu/2 = \Omega_\eps(\gamma d_i)$ triangles with a third
vertex in $C$. Since $|C| = O_\eps(d_i)$ and $\delta = \eps^c\gamma$ is sufficiently small, these
edges are red. Their endpoints lie in $S$, which has $O_\eps(d_i)$ vertices, so the red subgraph
they induce has average degree $\Omega_\eps(\gamma d_i)$. By \Lem{kcore}, it contains a core of
minimum degree $\Omega_\eps(\gamma d_i) \geq \delta |C|$. Thus, more vertices are added to $C$.
The loop does not terminate, completing the contradiction. 
\end{proof}

\subsection{Wrapping up} \label{sec:wrapup}

The proof of \Thm{main-contain} follows directly.

\begin{proof} Consider the family $\CC$ output by \mainproc$(G,\eps,\gamma)$. By \Clm{size},
the edge density of any $C \in \CC$ is $\Omega_{\eps,\gamma}(1)$. Consider any identifiable
$\eps$-triadic cohort $H$. By \Clm{goodseed} and \Lem{contain}, some output set of $\CC$
contains a $(1-\gamma)$-fraction of $V(H)$.
\end{proof}

We now look at the running time bound for \mainproc.  We use $t$ to denote the number of triangles in $G$,
$t_i$ to denote the triangle count incident to vertex $i$, and $R$ to be 
the running time of a triangle enumeration procedure. The analysis is pessimistic and just
for theoretical purposes, and the full proof is delegated to Appendix~\ref{app:runtime-thm}.

\begin{theorem} \label{thm:runtime} There is an implementation of \mainproc$(G,\eps,\gamma)$
that runs in time $O_{\eps,\gamma}((m + n + R)\log n + t \max_i \sqrt{t_i})$.
The storage is $O(m+n+t)$.
\end{theorem}
We note that triangle enumeration in sparse, real-world graphs is $O(m\alpha)$, where $\alpha$ is the degeneracy of the graph. 

\section{Empirical evaluation} \label{sec:prac}

We give a detailed empirical evaluation of \mainproc. 
We run our experiments on a 3.2 GHz Intel Xeon node with 128GB of memory. Our code is available at the anonymized repository \url{https://github.com/cluelessbasu/cohort-recovery}, and is written in C++, compiled with G++, and O3 optimizations.

\paragraph{Implementation differences:} 
We discuss important changes in the implemented algorithm. We forego the parameter $\gamma$
entirely, and use a greedy approach for growing the set $C$ in \contain. 
When a vertex $v$ is added to $C$, it potentially adds new triangles within $C$,
but also leads to new triangles that get cut by $C$. If the former is larger
than the latter, then we add $v$ to $C$. (The theoretical procedure in \contain\ uses
the threshold $\delta |C|^2$.) 

Edges are colored red according to the same criterion. For each edge, if the number
of incident triangles to $C$ is more than those outside $C$, we color the edge red.
Also, instead of computing various core numbers, we simply take the 2-core of the red edges.
We observed that this led to almost no loss in the output densities, but led to a significantly simpler algorithm.

Thus, we only have a single tunable parameter of $\eps$ for our algorithm; for our experimental results, $\eps$ is set to $0.1$. We elaborate on the role of $\gamma$ below, and provide an expanded discussion on $\eps$ in Section~\ref{sec:findings} and Figure~\ref{fig:ablation}.

While the theoretical construction stores the entire triangle list,
a common technique is to locally recompute triangles when required. We find this trick to significant reduce the running time in practice.

\subsection{Experimental setup} \label{sec:experiments}

We compare \mainproc\ with a number of existing algorithms.

\begin{asparaenum}
    \item \emph{RTRExtractor}: A triangle based dense cluster extraction algorithm~\cite{BaPeQi+24}. It produces disjoint subgraphs. We use the recommended setting of $\eps = 0.1$.
    \item \emph{CoDeSEG}: A recent algorithm~\cite{XLP+25} that uses structural entropy to detect overlapping clusters in graphs. We use the default settings for the algorithm.     
    \item \emph{Nucleus}: A hierarchical dense subgraph algorithm that produces potentially overlapping clusters~\cite{SaSePi+15}. We use the (2,3)-nuclei, which is also called the $k$-truss decomposition, and evaluate all clusters in its output hierarchy.
    \item \emph{Louvain}: One of the most popular community detection algorithms~\cite{Louvain}, it produces disjoint hierarchical clusters using a heuristic for modularity maximization. We use the default parameters in their source code~\cite{LouvainImp}.
    \item \emph{Leiden}: A refinement of modularity-based community detection designed to produce well-connected communities~\cite{Leiden}. 
    \item \emph{AGM}: A PageRank based algorithm for producing overlapping clusters in graphs, with an emphasis on distributed computing~\cite{AGM12}. We set the sole conductance parameter to $1000.$
\end{asparaenum}
\begin{table}
\centering

\begin{tabular}{|l|r|r|r|r|r||c|}
\hline
Dataset & $|V|$ & $|E|$ & CC & $\alpha$ & G & Rtime \\
\hline
euemail & 1.00 k & 16.06 k & 0.40 & 34 & \ding{51} & 204 ms \\
soc-ham & 2.43 k & 16.63 k & 0.54 & 24 & \ding{55} & 95 ms \\
caAstroPh & 18.77 k & 198.05 k & 0.63 & 56 & \ding{55} & 591 ms \\
Ph-Citations & 27.77 k & 352.28 k & 0.31 & 37 & \ding{55} & 554 ms \\
Epinions & 75.89 k & 405.74 k & 0.14 & 67 &\ding{55} & 911 ms \\
Slashdot & 82.17 k & 504.23 k & 0.06 & 55 &\ding{55} & 685 ms \\
dblp & 317.08 k & 1.05 M & 0.63 & 113 &\ding{51} & 683 ms \\
amazon & 334.86 k & 925.87 k & 0.40 & 6 &\ding{51} & 634 ms \\
Berkstan & 685.23 k & 6.65 M & 0.60 & 201 & \ding{55} & 5m 51s \\
hollywood & 1.07 M & 56.31 M & 0.77 & 2.2k &\ding{55} & 58m 48s \\
youtube & 1.13 M & 2.99 M & 0.08 & 51 &\ding{51} & 3.27 s \\
pokec & 1.63 M & 22.30 M & 0.11 & 47 &\ding{55} & 38.61 s \\
skitter & 1.70 M & 11.10 M & 0.26 & 111 &\ding{55} & 12.93 s \\
wiki & 1.79 M & 25.44 M & 0.27 & 99 &\ding{51} & 54.56 s \\
large-dblp & 1.82 M & 8.34 M & 0.63 & 286 &\ding{55} & 10.40 s \\
orkut & 3.07 M & 117.19 M & 0.17 & 253 &\ding{51} & 7m 39s \\
cit-Patents & 3.77 M & 16.52 M & 0.08 & 64 &\ding{55} & 19.26 s \\
livejournal & 4.00 M & 34.68 M & 0.28 & 360 &\ding{51} & 1m 5s \\
\hline
\end{tabular}

\caption{Datasets used in our experiments: vertices, edges, clustering coefficients and degeneracy. Datasets with annotated ground truth clusters are marked with a check(\ding{51}) while the others bear a cross(\ding{55}). In the last column, we mention the runtime for \mainproc{} excluding I/O. }
\label{tab:dataset_stats}
\end{table}

We experiment on a wide variety of public datasets, listed in \Tab{dataset_stats} 
(obtained from SNAP~\cite{snapnets} and the network repository~\cite{networkrepository}).
We consider all datasets to be undirected and remove self loops. 
In \Tab{gt-membership}, we give details about
curated ground truth clusters available for seven of
the datasets.

\begin{table}[ht]
    \centering
    \begin{tabular}{l|cc}
        \hline
        \textbf{Dataset} & \textbf{Frac. clustered} & \textbf{Frac. $> 1$ cluster} \\
        \hline
        wiki & 1.0 & 0.55  \\
        euemail & 1.0 & 0.0 \\
        amazon & 0.95 & 0.91 \\
        dblp & 0.82 & 0.35 \\
        orkut & 0.76 & 0.71  \\
        livejournal & 0.29 & 0.19 \\
        youtube & 0.05 & 0.02 \\
        \hline
    \end{tabular}
       
    \caption{A look at our labeled datasets. \emph{Frac. clustered} gives the fraction of the graph's vertices that appear in some labeled ground truth cluster. The last column, \emph{Frac. $> 1$ cluster}, gives the fraction of vertices that appear in multiple ground truth clusters. Observe that, barring two datasets, overlap is a significant aspect of the ground truth cluster structure.
    }  \label{tab:gt-membership}  
\end{table}

Throughout the rest of this section, we use the abbreviation "Coh"  for \mainproc{}. 
\begin{table*}
\resizebox{\textwidth}{!}{%
\begin{tabular}{|l|ccccccc|ccccccc|}
\hline
\textbf{Dataset} &\multicolumn{7}{c|}{\textbf{0.5}} & \multicolumn{7}{c|}{\textbf{0.8}} \\
\cline{2-15}
 & Coh & RTR & Cod & Nuc & Lou & Lei & AGM & Coh & RTR & Cod & Nuc & Lou & Lei & AGM \\
\hline
Ph-Citations
& \cellcolor{green!25}{0.585} & 0.317 & 0.004 & 0.118 & 0.054 & 0.056 & 0.000
& \cellcolor{green!25}{0.348} & 0.247 & 0.001 & 0.070 & 0.006 & 0.007 & 0.000 \\

Slashdot
& 0.020 & 0.006 & 0.000 & \cellcolor{green!25}{0.022} & 0.016 & 0.015 & 0.000
& \cellcolor{green!25}{0.010} & 0.005 & 0.000 & 0.009 & 0.001 & 0.001 & 0.000 \\

amazon
& \cellcolor{green!25}{0.489} & 0.096 & 0.063 & 0.355 & 0.248 & 0.245 & 0.002
& \cellcolor{green!25}{0.225} & 0.089 & 0.011 & 0.164 & 0.035 & 0.035 & 0.000 \\

caAstroPh
& \cellcolor{green!25}{0.654} & 0.472 & 0.013 & 0.194 & 0.266 & 0.273 & 0.038
& \cellcolor{green!25}{0.618} & 0.468 & 0.006 & 0.149 & 0.078 & 0.078 & 0.005 \\

cit-Patents
& \cellcolor{green!25}{0.116} & 0.051 & 0.007 & 0.119 & 0.029 & 0.028 & 0.000
& \cellcolor{green!25}{0.046} & 0.032 & 0.000 & 0.028 & 0.001 & 0.000 & 0.000 \\

dblp
& \cellcolor{green!25}{0.454} & 0.184 & 0.022 & 0.346 & 0.276 & 0.271 & 0.002
& \cellcolor{green!25}{0.396} & 0.184 & 0.004 & 0.252 & 0.071 & 0.070 & 0.000 \\

euemail
& \cellcolor{green!25}{0.588} & 0.350 & 0.000 & 0.085 & 0.063 & 0.061 & 0.000
& \cellcolor{green!25}{0.412} & 0.335 & 0.000 & 0.045 & 0.000 & 0.000 & 0.000 \\

hollywood
& \cellcolor{green!25}{0.833} & 0.770 & 0.020 & OOM & 0.322 & 0.326 & 0.138
& \cellcolor{green!25}{0.832} & 0.769 & 0.007 & OOM & 0.291 & 0.294 & 0.049 \\

large-dblp
& \cellcolor{green!25}{0.454} & 0.231 & 0.019 & 0.258 & 0.215 & 0.214 & 0.003
& \cellcolor{green!25}{0.415} & 0.231 & 0.006 & 0.201 & 0.068 & 0.068 & 0.001 \\

livejournal
& \cellcolor{green!25}{0.331} & 0.178 & 0.001 & 0.130 & 0.048 & 0.048 & 0.006
& \cellcolor{green!25}{0.242} & 0.164 & 0.000 & 0.083 & 0.009 & 0.009 & 0.002 \\

orkut
& \cellcolor{green!25}{0.349} & 0.242 & 0.000 & 0.092 & 0.003 & 0.003 & 0.001
& \cellcolor{green!25}{0.287} & 0.161 & 0.000 & 0.058 & 0.000 & 0.000 & 0.000 \\

pokec
& \cellcolor{green!25}{0.120} & 0.078 & 0.000 & 0.078 & 0.002 & 0.002 & 0.000
& \cellcolor{green!25}{0.091} & 0.066 & 0.000 & 0.041 & 0.000 & 0.000 & 0.000 \\

skitter
& \cellcolor{green!25}{0.075} & 0.026 & 0.000 & 0.032 & 0.018 & 0.018 & 0.000
& \cellcolor{green!25}{0.028} & 0.016 & 0.000 & 0.007 & 0.001 & 0.001 & 0.000 \\

Epinions
& \cellcolor{green!25}{0.060} & 0.029 & 0.000 & 0.030 & 0.017 & 0.017 & 0.000
& \cellcolor{green!25}{0.044} & 0.026 & 0.000 & 0.016 & 0.001 & 0.001 & 0.000 \\

soc-ham
& \cellcolor{green!25}{0.586} & 0.452 & 0.079 & 0.249 & 0.290 & 0.308 & 0.024
& \cellcolor{green!25}{0.484} & 0.435 & 0.055 & 0.215 & 0.155 & 0.157 & 0.016 \\

Berkstan
& \cellcolor{green!25}{0.317} & 0.167 & 0.040 & 0.081 & 0.110 & 0.112 & 0.029
& \cellcolor{green!25}{0.228} & 0.128 & 0.018 & 0.041 & 0.029 & 0.030 & 0.005 \\

wiki
& \cellcolor{green!25}{0.083} & 0.051 & 0.000 & 0.010 & 0.003 & 0.003 & 0.000
& \cellcolor{green!25}{0.061} & 0.044 & 0.000 & 0.006 & 0.001 & 0.001 & 0.000 \\

youtube
& 0.021 & 0.005 & 0.000 & \cellcolor{green!25}{0.024} & 0.014 & 0.014 & 0.000
& \cellcolor{green!25}{0.011} & 0.004 & 0.000 & 0.008 & 0.000 & 0.000 & 0.000 \\

\hline
\end{tabular}%
}
\caption{The coverage of \mainproc{} and competing methods across a variety of datasets, restricting to clusters with at least 5 vertices. We consider coverage thresholds of $0.5$ and $0.8$.}
\label{tab:cohorts-coverage}
\end{table*}
\subsection{Main findings} \label{sec:findings}

\paragraph{High coverage of \mainproc:} Arguably the main achievement of \mainproc{} is the ability to cover large portions of a graph with dense subgraphs. 
For every algorithm output, we remove sets that with at most $5$ vertices, as these are too small to be of interest. 
We measure
\emph{coverage at density $\rho$}. Consider a collection of sets $\CC$ where
every $C \in \CC$ induces a subgraph of density at least $\rho$. The coverage of $\CC$
is the fraction of vertices of $G$ contained in $\CC$. We consider coverage at densities
$0.5$ and $0.8$. We then
consider all sets of density at least $\rho$, and compare the coverage over various
algorithms (and repeat this for all the datasets). For $\rho = 0.5$, coverage for \mainproc\ is highest for all but three datasets, where it is a close second. 
In eleven of the eighteen chosen datasets, coverage at 0.5 for \mainproc\ is above 0.3. At $\rho =0.8$, \mainproc\ has the highest coverage on all eighteen datasets. We refer the reader to ~\Tab{cohorts-coverage} for the full data on coverage. We note that Nucleus runs into an OOM (out of memory) error on the hollywood dataset on our machine.

We note that \mainproc{} finishes on our machine on all of our datasets in less than an hour on all our datasets. In the last column of \Tab{dataset_stats}, we look at the total runtime of the algorithm excluding I/O.  We note that even for very large graphs, the total time taken is typically seconds to minutes. Section~\ref{app:runtime-exp} considers a detailed break down of the time taken by each component.

\paragraph{Capturing ground truth information:} Consider the datasets
with ground truth information: these are the seven datasets in \Tab{gt-membership}. Of these, we focus on the six datasets (except euemail) that have overlapping ground-truth clusters. We stress that this information can be quite noisy, and
is imperfect. Typically, these networks have category labels, and the connected components of vertices in some category are considered to be the "ground truth".  In some cases, like livejournal, the communities are real groups on social networks. Still, these ground truth clusters have information,
and we wish to see how well \mainproc{} can recover this information. We measure a ground-truth-oriented notion of \emph{MaxPrec}. Let $\cD$ and $\cG$ be the detected and ground-truth cluster collections, respectively. We first retain ground-truth clusters of size at least five,
$\cG_5=\{G\in\cG:|G|\geq 5\}$, and let $U_5=\bigcup_{G\in\cG_5}G$ be their labeled vertex universe. We retain raw detected clusters of size at least five, intersect each with $U_5$, and discard empty intersections: $\widetilde{\cD}_5 = \{D\cap U_5:D\in\cD,\ |D|\geq 5,\ D\cap U_5\neq\emptyset\}$.
The reported score is
\begin{align}
    \text{MaxPrec}(\cD,\cG)
    = \frac{1}{|\cG_5|}\sum_{G_j\in\cG_5}
      \max_{\widetilde D_i\in\widetilde{\cD}_5}
      \frac{|\widetilde D_i\cap G_j|}{|\widetilde D_i|}. \label{eq:maxprec}
\end{align}
Thus, for each retained ground-truth cluster, we find the retained intersection with the largest precision relative to that ground-truth cluster, and then average over ground-truth clusters. The size threshold is applied before intersection; a nonempty intersection remains eligible even if it has fewer than five vertices. Refer to \Tab{gt-bc}. We note that in five of our six datasets, \mainproc\ emerges the winner.

\begin{table}
    \centering
\small{
\begin{tabular}{lcccccc}
\toprule
\textbf{Method} & \textbf{youtube} & \textbf{dblp} & \textbf{amazon} & \textbf{wiki} & \textbf{livejournal} & \textbf{orkut} \\
\midrule
Cohorts & 0.296 & \cellcolor{green!25}{0.433} & \cellcolor{green!25}{0.400} & \cellcolor{green!25}{0.278} & \cellcolor{green!25}{0.653} & \cellcolor{green!25}{0.144} \\
Nucleus & 0.320 & 0.364 & 0.329 & 0.133 & 0.394 & 0.097 \\
CoDeSEG & 0.078 & 0.143 & 0.215 & 0.101 & 0.087 & 0.010 \\
AGM & 0.251 & 0.030 & 0.029 & 0.053 & 0.118 & 0.009 \\
RTREx & 0.132 & 0.293 & 0.195 & 0.275 & 0.530 & 0.128 \\
Louvain & \cellcolor{green!25}{0.684} & 0.374 & 0.325 & 0.103 & 0.388 & 0.014 \\
Leiden & 0.680 & 0.368 & 0.325 & 0.110 & 0.400 & 0.012 \\
\bottomrule
\end{tabular}}
\caption{Ground-truth-oriented MaxPrec across methods. 
Cells marked in green are the best for that dataset (per column).}\label{tab:gt-bc}
\end{table}

\paragraph{Ablation studies: }
We vary the parameter $\eps\in[0.05, 0.5]$ at increments of $0.05$, and run for all datasets. 
When we set $\eps \leq 0.01$, the edge removal step of \mainproc\ is ineffective, and the output sets have low density. 
(In the limit when $\eps = 0$, no edge is ever deleted, and the entire graph is returned as a single cluster.)
In \Fig{ablation}, we show the coverage with density $0.8$ across the values of $\eps$.
Coverage drops significantly as $\eps$ increases beyond $0.2$. The optimal value is between $0.1$ and $0.2$ for different
datasets. Overall, the setting of $\eps = 0.1$
works well for all datasets, and small changes of $\eps$ have almost no effect on the output. 
In Appendix~\ref{sec:ablation-app}, we also show results for coverage at density $0.5$.

\begin{figure}
    \centering
    \includegraphics[width=0.5\linewidth]{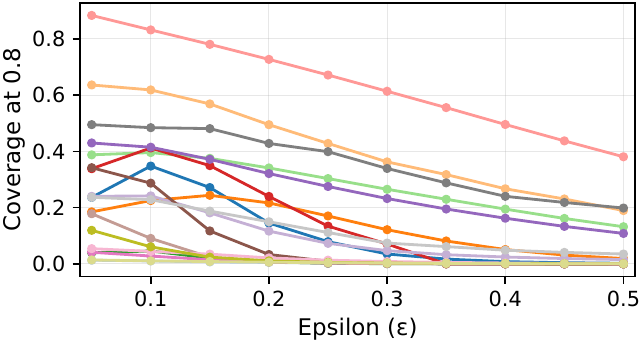}
    \includegraphics[width=0.7 \linewidth]{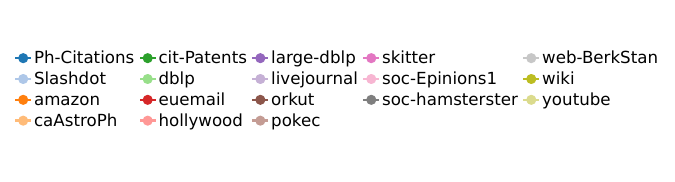}
        \vspace{-0.5cm}
    \caption{Coverage at 0.8 at different $\eps$. We observe that the coverage broadly decays with increasing $\eps$, but the maximum is typically between 0.05 and 0.2 for different datasets.}
    \label{fig:ablation}

\end{figure}

\paragraph{Recovery in planted models: }
We consider Stochastic Block Models with disjoint and overlapping clusters, so that
ground truth is known.
We test on 30 independently generated graph instances with blocks of 100 vertices. In the disjoint case, there are 20 blocks.
In the overlapping setting, we have 26 planted blocks in order, and each block overlaps with the "adjacent" (in the ordering)
blocks by 25 vertices. Two vertices in a block form an edge in with independently probability $p_{\mathrm{in}}$ (set to 0.7 or above).
All other pairs form an edge with probability $p_{\mathrm{out}}$, fixed to $0.1$.

For space reasons, we provide results in the appendix, in \Tab{sbm-app}.
For each output cluster of size at least 5, we compute the MaxPrec. We also compute the recall for every block.
For sake of space, we only show results on Nucleus and CoDeSEG. (Other methods are comparable, or significantly worse.)
For the simple disjoint case, both \mainproc\ and Nucleus have near perfect recall and MaxPrec. CoDeSEG has significantly lower
MaxPrec, showing that many ground truth clusters are not captured by its output.
For the overlapping case, the performance of CoDeSEG is much worse. We note that Nucleus has near perfect recall,
better than \mainproc. But the MaxPrec values are much lower, again showing that many ground truth clusters
are not covered properly.

\paragraph{Examples of dense subgraphs:} We apply our algorithm to a labeled citation network of computer science papers, sourced from AMiner~\cite{aminer}. 
In \Tab{aminer-tables} of the appendix, we show an example of three intersecting output sets. Each cluster is on a coherent topic in reconfigurable systems, but they also overlap significantly. We found tens of thousands of clusters on this dataset. This is a striking example of the importance of allowing overlaps. Such clusters would be incomplete in a true partition.

\printbibliography
\clearpage
\appendix
\section{Appendix}
 
\subsection{\texorpdfstring{Proofs from \Sec{prelims}}{Proofs from Preliminaries}} \label{app:prelims}
 
We give the proof of \Lem{kcore}.
 
\begin{proof} Let $S$ denote a set of vertices whose induced subgraph has average degree at least $r$.
Hence, $S$ contains at least $r|S|/2$ edges.
Consider the Matula-Beck procedure that repeatedly removes the vertex of minimum degree. Let us
stop it before it removes a vertex of degree at least $r/2$. So at this point, every vertex removal
removes strictly less than $r/2$ edges from the graph. There can be strictly less than $r|S|/2$
edges removed from the subgraph induced on $S$. Hence, there is at least some edge remaining
at this point of the procedure, so there is a remaining subgraph with minimum degree at least $r/2$.
\end{proof}
 
We give the proof of \Clm{reverse}.
 
\begin{proof} 
\begin{align}
\EX[X] = & \Pr[X \leq \EX[X]/2] \EX[X | X \leq \EX[X]/2] \\
& + \Pr[X > \EX[X]/2] \EX[X | X > \EX[X]/2] \\
\leq & \EX[X]/2 + \Pr[X > \EX[X]/2] M
\end{align}
We rearrange to complete the proof.
\end{proof}
 
\subsection{\texorpdfstring{Proofs from \Sec{contain}}{Proofs from Cohort Containment}} \label{app:contain}
 
We give the proof of \Clm{i-degree}.
 
\begin{proof} First, consider $G'$ from \mainproc. The vertex $i$ has non-zero degree in $G'$, since all isolated vertices
are removed in \mainproc.
Each edge $(i,v)$ must participate
in at least $\eps \max(d_i, d_v)$ triangles. The number of these triangles
is at most $\min(d_i, d_v)$. Hence, $d_v \leq d_i/\eps$ and $v$ (indeed all neighbors
of $i$ in $G'$) is in $G''$. The existence of at least $\eps d_i$ triangles incident to $i$
means that the degree of $i$ in $G''$ is at least $\eps d_i$. Every edge
incident to $i$ participates in at least $\eps d_i$ triangles in $G'$; the other
vertices in the triangle have degree at most $d_i/\eps$, and are therefore in $G''$.
So all these (at least) $\eps d_i$ triangles are in $G''$.
\end{proof}
 
\begin{figure}[ht!]
    \centering
    \includegraphics[width = \linewidth]{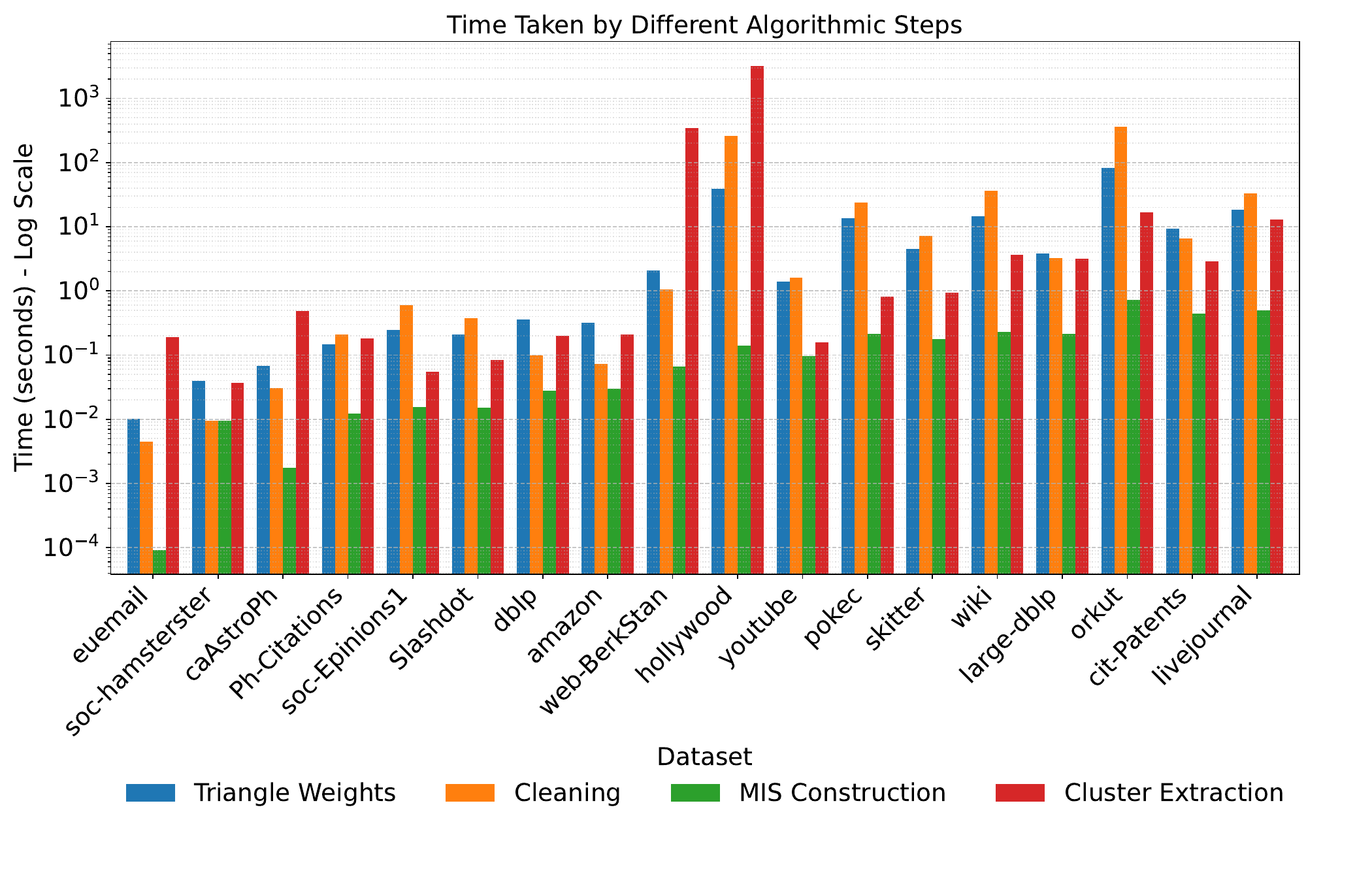}
    \caption{The time taken by various crucial steps of Cohorts.}
    \label{fig:runtime}
\end{figure}

\subsection{Runtime Analysis}\label{app:runtime-thm}
 
We prove \Thm{runtime}.
 
\begin{proof} The cleaning operation in \Step{clean} has a standard procedure and 
analysis~\cite{GuRoSe14,BBS24,BaPeQi+24}. We first enumerate all triangles in $R$ time,
and store them in a data structure indexed by vertex and edge. For every edge $(u,v)$,
we can keep track of $t_{u,v}/\max(d_u, d_v)$ (where $t_{u,v}$ is the triangle count of edge $(u,v)$)
in a priority queue, and update every time a new
triangle is found. Finally, we remove edges with a triangle count below $\eps \max(d_u, d_v)$,
and update the data structure by also deleting triangles participating on the removed edge.
The total running time is $O(R + (m+n+t)\log n) = O((m+n+R)\log n)$. We store
all triangles, so the storage is $O(m+n+t)$.
 
The interesting analysis is bounding the running time of \contain$(i,G',\eps,\gamma)$,
for every $i$ with non-zero degree in $G'$. We will show that the running time is $O(d^3_i)$.
We maintain a data structure that stores all triangles incident to at least one vertex in $C$.
These triangles will be indexed by vertex and edge. When a vertex enters $C$, we enumerate
all triangles it participates in to update this data structure. Every vertex has degree at most $d_i/\eps^3$,
so the time to find all these triangles is $O_\eps(d^2_i)$. By \Clm{size}, there are $O_{\eps,\gamma}(d_i)$ vertices in $C$. So the total time for all this enumeration is $O_{\eps,\gamma}(d^3_i)$.
 
Recall that we store vertices and edges in a priority queue keyed with the number of triangles incident
to $C$. So it is straightforward to find vertices to add to $C$, or to color edges red. All the
red edges are stored in an auxiliary graph, and we can compute the core decomposition of this
graph in linear time. The number of edges is at most the number of triangles seen thus far,
which (by the previous paragraph) is $O_{\eps,\gamma}(d^3_i)$. So the core decomposition
takes the same time. Observe that each time a core of degree at least $\delta|C|$ is discovered,
at least $\Omega_{\eps, \gamma}(d_i)$ vertices are added to $C$ (recall $\delta = \Theta_{\eps,\gamma}(1)$). Since $|C| = O_{\eps,\delta}(d_i)$, this operation can only happen $O_{\eps,\gamma}(1)$ times.
So overall, the running time is $O_{\eps,\gamma}(d^3_i)$.
Note that, by the analysis in the proof of \Clm{size}, $i$ participates in $\Omega_{\eps}(d^2_i)$
triangles. Hence, the running time will be $O_{\eps\gamma}((t^{3/2}_i))$.
 
We now sum over all $i$ to get the total running time for all calls to \contain. 
We get $\sum_{i \in V} t^{3/2}_i \leq \max_i \sqrt{t_i} \sum_i t_i = t \max_i \sqrt{t_i}$.
That completes the overall proof.
\end{proof}
 
\subsection{Empirical running time breakdown}\label{app:runtime-exp}
We discuss more on the empirical running time, by 
breaking down the time taken by each component of \mainproc. \Fig{runtime}
gives the results across all the datasets.
Typically, the initial cleaning takes the most time, but in some cases, cluster extraction is the biggest component. The two datasets that take disproportionately more time for their size (Berkstan and hollywood) are ones with notably large clusters.
We note that hollywood has larger clusters produced compared to the other datasets, which may explain the higher running time.
 
\subsection{Examples of dense subgraphs} \label{sec:app-example}
 
As described in the main body, we apply \mainproc\ to a labeled citation network of computer science papers, obtained from AMiner~\cite{aminer}. 
We get tens of thousands on clusters, many of which are semantically meaningful.
\Tab{aminer-tables} gives an example of three intersecting output sets. 
Each cluster is on a coherent topic in reconfigurable systems, which has non-trivial
overlap denoted in purple).

\begingroup
\setlength{\intextsep}{0pt}
\begin{table}[H]
    \centering
    \scriptsize 
    \setlength{\tabcolsep}{2pt} 
    \renewcommand{\arraystretch}{1.05}
 
       \vspace{0pt} 
        \centering
        \begin{tabular}{|p{0.38\textwidth}|l|}
            \hline
            \textbf{Publication Name} & \textbf{Venue} \\
            \hline
            HW/SW partitioning \& code gen. of embedded control apps & CODES '02 \\ \hline
            A DAG-based design approach for reconfig. VLIW processors & DATE '99 \\ \hline
            Instr. gen. for hybrid reconfig. systems & TODAES '02 \\ \hline
            Synthesis of custom processors based on extensible platforms & ICCAD '02 \\ \hline
            The MOLEN Polymorphic Processor & IEEE TC '04 \\ \hline
            Instr. gen. and regularity extraction for reconfig. processors & CASES '02 \\ \hline
            Automatic gen. of app-specific processors & CASES '03 \\ \hline
            \rowcolor{violet!20}
            Processor Acceleration Through Automated Instr. Set Customization & MICRO '03 \\ \hline
            Integer linear programming approach for identifying instr.-set extensions & CODES+ '05 \\ \hline
            \rowcolor{violet!20}
            Scalable custom instr. identification for instr.-set extensible processors & CASES '04 \\ \hline
            \rowcolor{violet!20}
            App-specific instr. gen. for configurable processor arch. & FPGA '04 \\ \hline
            Automatic selection of app-specific instr.-set extensions & CODES+ '06 \\ \hline
            Linear complexity algo. for gen. of MISO instr. of variable size & SAMOS '07 \\ \hline
            Linear complexity algo. for auto. gen. of convex MIMO instr. & ARC '07 \\ \hline
            \rowcolor{violet!20}
            Modern dev. methods \& tools for embedded reconfig. systems: A survey & Integ. '10 \\ \hline
        \end{tabular}
        \begin{tabular}{|p{0.38\textwidth}|l|}
            \hline
            \textbf{Publication Name} & \textbf{Venue} \\
            \hline
            Designing domain-specific processors & CODES '01 \\ \hline
            Xtensa: A Configurable and Extensible Processor & Micro '00 \\ \hline
            Efficient instr. encoding for auto. instr. set design of config. ASIPs & ICCAD '02 \\ \hline
            Characterizing embedded apps for instr.-set extensible processors & DAC '04 \\ \hline
            \rowcolor{violet!20}
            Processor Acceleration Through Automated Instr. Set Customization & MICRO '03 \\ \hline
            Instr.-set customization for real-time embedded systems & DATE '07 \\ \hline
            Satisfying real-time constraints with custom instr. & CODES+ '05 \\ \hline
            \rowcolor{violet!20}
            Scalable custom instr. identification for instr.-set extensible processors & CASES '04 \\ \hline
            \rowcolor{violet!20}
            App-specific instr. gen. for configurable processor arch. & FPGA '04 \\ \hline
            Evaluating design trade-offs in customizable processors & DAC '09 \\ \hline
        \end{tabular}
 
        \vspace{0.5em}
 
        \begin{tabular}{|p{0.38\textwidth}|l|}
            \hline
            \textbf{Publication Name} & \textbf{Venue} \\
            \hline
            Kernel scheduling in reconfig. computing & DATE '99 \\ \hline
            Framework for Scheduling \& Context Allocation in Reconfig. Computing & ISSS '99 \\ \hline
            Temporal Partitioning \& Scheduling for Reconfig. Computing & FCCM '98 \\ \hline
            Arch. Synthesis Techniques for Dynamically Reconfig. Logic & FPL '96 \\ \hline
            Framework for reconfig. computing: task scheduling \& context mgmt & TVLSI '01 \\ \hline
            Config. mgmt in multi-context reconfig. systems... & ISSS '00 \\ \hline
            \rowcolor{violet!20}
            Modern dev. methods \& tools for embedded reconfig. systems: A survey & Integ. '10 \\ \hline
        \end{tabular}
\caption{Intersection analysis showing coherent sub-fields in the overlap: reconfigurable architecture generation (top), instruction set customization (middle), and task scheduling (bottom). Rows highlighted in violet are papers in the intersection.}\label{tab:aminer-tables}
    
\end{table}

\begin{figure}[H]
    \centering
    \includegraphics[width=0.46\linewidth, height=0.28\textheight, keepaspectratio]{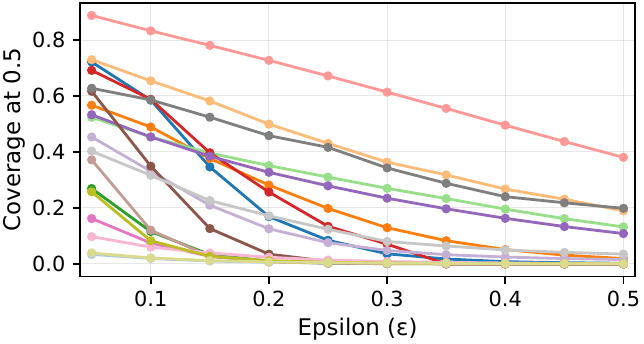}
    \includegraphics[width=0.46\linewidth, height=0.28\textheight, keepaspectratio]{newfigs/coverage_trends_regular_mis_rho0.8.pdf}
    \includegraphics[width=0.7\linewidth]{newfigs/coverage_trends_regular_mis_index.pdf}
    \caption{Coverage at 0.5 (top) and 0.8 (bottom) at different thresholds for varying $\eps$. Both show sharp drops in coverage after $\eps=0.2$, but the drop for coverage at 0.5 is noticeably higher.}
    \label{fig:ablation-app}
\end{figure}
\endgroup

\begin{table*}[b!]
\centering
\setlength{\tabcolsep}{3.2pt}
\resizebox{\textwidth}{!}{
\begin{tabular}{l|ccc|ccc|ccc}
\toprule
& \multicolumn{3}{c|}{Coh}
& \multicolumn{3}{c|}{\textbf{Nucleus}}
& \multicolumn{3}{c}{\textbf{CoDeSEG}} \\
$p_{\mathrm{in}}$
& \#Output$\geq 5$ & Recall & MaxPrec
& \#Output$\geq 5$ & Recall & MaxPrec
& \#Output$\geq 5$ & Recall & MaxPrec \\
\midrule
0.70
& $90.03 \pm 3.16$ & $1.000 \pm 0.000$ & $1.000 \pm 0.000$
& $96.57 \pm 6.38$ & $1.000 \pm 0.000$ & $1.000 \pm 0.000$
& $5.13 \pm 1.07$ & $0.869 \pm 0.028$ & $0.216 \pm 0.050$ \\
0.80
& $72.40 \pm 2.94$ & $1.000 \pm 0.000$ & $1.000 \pm 0.000$
& $81.87 \pm 6.00$ & $1.000 \pm 0.000$ & $1.000 \pm 0.000$
& $6.50 \pm 1.17$ & $0.867 \pm 0.026$ & $0.283 \pm 0.055$ \\
0.90
& $54.70 \pm 2.49$ & $1.000 \pm 0.000$ & $1.000 \pm 0.000$
& $60.93 \pm 4.49$ & $1.000 \pm 0.000$ & $1.000 \pm 0.000$
& $6.93 \pm 1.76$ & $0.875 \pm 0.025$ & $0.300 \pm 0.081$ \\
\bottomrule
\end{tabular}
}
\resizebox{\textwidth}{!}{
\begin{tabular}{l|ccc|ccc|ccc}
\toprule

$p_{\mathrm{in}}$
& \#Output$\geq 5$ & Recall & MaxPrec
& \#Output$\geq 5$ & Recall & MaxPrec
& \#Output$\geq 5$ & Recall & MaxPrec \\
\midrule
0.70
& $101.00 \pm 3.05$ & $0.927 \pm 0.006$ & $1.000 \pm 0.000$
& $34.47 \pm 2.27$ & $1.000 \pm 0.000$ & $0.434 \pm 0.042$
& $5.80 \pm 1.40$ & $0.894 \pm 0.026$ & $0.247 \pm 0.059$ \\
0.80
& $81.60 \pm 3.57$ & $0.934 \pm 0.005$ & $1.000 \pm 0.000$
& $39.47 \pm 2.84$ & $1.000 \pm 0.000$ & $0.446 \pm 0.047$
& $6.60 \pm 1.63$ & $0.901 \pm 0.020$ & $0.283 \pm 0.062$ \\
0.90
& $62.53 \pm 2.91$ & $0.961 \pm 0.003$ & $1.000 \pm 0.000$
& $41.90 \pm 2.44$ & $1.000 \pm 0.000$ & $0.457 \pm 0.045$
& $6.73 \pm 1.87$ & $0.902 \pm 0.019$ & $0.289 \pm 0.072$ \\
\bottomrule
\end{tabular}
}
\caption{Recovery at fixed $p_{\mathrm{out}}=0.10$ on disjoint (top) and overlapping (bottom) SBMs. Values are mean $\pm$ sample standard deviation over 30 instances. \#Output$\geq 5$ is the number of output clusters of size at least five per instance. Recall and MaxPrec are averaged over planted blocks within each randomized instance, then averaged across instances.}
\label{tab:sbm-app}
\end{table*}

This is a striking example of the importance of allowing overlaps. Such clusters would be incomplete in a true partition.

\subsection{Studies with varying \texorpdfstring{$\eps$}{epsilon}} \label{sec:ablation-app}
 
(For convenience, we restate some text from the main body.)
We present the results for coverage at varying $\eps$ values.
We vary the parameter $\eps\in[0.05, 0.5]$ at increments of $0.05$, and run for all datasets. 
When we set $\eps \leq 0.01$, the edge removal step of \mainproc\ is ineffective, and the output sets have low density. 
(In the limit when $\eps = 0$, no edge is ever deleted, and the entire graph is returned as a single cluster.)
In \Fig{ablation-app}, we show the coverage with density $0.8$ and $0.5$ across the values of $\eps$.
Across all plots, coverage drops significantly as $\eps$ increases beyond $0.2$. 
For density $0.5$, it appears to be mostly monotonic, and smaller $\eps$ is generally better.
But for density $0.8$, there are numerous instances where coverage increases
by small increases (between $0.1$ and $0.2$) in $\eps$.
Overall, the setting of $\eps = 0.1$
works well for all datasets, and small changes of $\eps$ have almost no effect on the output. 
So we simply set $\eps = 0.1$ for all of our main experiments.

\subsection{Results on SBMs} \label{sec:sbm-app}

(For convenience, we restate part of the main-text discussion.)

We evaluate on Stochastic Block Models (SBMs) with known ground truth under both disjoint and overlapping cluster structure. Each setting uses 30 independently generated graphs with planted blocks of size 100.

In the disjoint setting, we plant 20 blocks for a total of 2000 vertices.
In the overlapping setting, we plant 26 blocks in sequence, where each block overlaps with its adjacent blocks by 25 vertices, adding up to 1975 vertices.
For any pair of vertices within the same planted block, an edge is included independently with probability $p_{\mathrm{in}} \ge 0.7$; all other pairs connect independently with probability $p_{\mathrm{out}}=0.1$.

Table~\ref{tab:sbm-app} reports results. 
Values are mean $\pm$ sample standard deviation over 30 graph instances. 
\#Output$\geq 5$ denotes the number of output clusters of size at least 5. 
For evaluation, we compute MaxPrec (Eq.~\ref{eq:maxprec}) for each output cluster, and recall (fraction recovered by its best-matching output cluster); we then average these quantities at the instance level and report their mean across instances. All metrics are computed restricting to output clusters with at least 5 vertices.

For space, we report Nucleus and CoDeSEG (other methods are comparable or weaker). 
On disjoint SBMs, \mainproc\ and Nucleus both achieve near-perfect recall and MaxPrec, while CoDeSEG has substantially lower MaxPrec, indicating that many planted blocks are not cleanly recovered. 
On overlapping SBMs, CoDeSEG degrades further. 
Nucleus attains near-perfect recall, slightly above \mainproc, but its much lower MaxPrec indicates that many output clusters mix multiple planted blocks rather than isolating them well. We also observe output redundancy for both \mainproc\ and Nucleus, reflected by \#Output$\geq 5$; for \mainproc, this count decreases as $p_{\mathrm{in}}$ increases in both disjoint and overlapping settings.

Coding agents were used to generate and run the simulations for the planted models.

\end{document}